\documentclass[conference,final]{IEEEtran}
\IEEEoverridecommandlockouts
\usepackage{cite}
\usepackage{amsmath,amssymb,amsfonts}
\usepackage{algorithmic}
\usepackage{graphicx}
\usepackage{textcomp}
\usepackage{orcidlink}
\usepackage{wrapfig}

\usepackage{bm}
\usepackage{enumitem}
\usepackage{cleveref}
\usepackage{xcolor}
\usepackage{comment}
\usepackage{subcaption}
\usepackage{amsthm}

\newtheorem{theorem}{Theorem}
\newtheorem{property}{Property}
\newtheorem{definition}{Definition}
\newtheorem{example}{Example}

\newtheorem{assumption}{Assumption}
\newtheorem{lemma}{Lemma}
\newtheorem{proposition}{Proposition}

\Crefname{property}{Property}{Properties}
\Crefname{assumption}{Assumption}{Assumptions}

\allowdisplaybreaks

\usepackage{tikz}

\newcommand\copyrighttext{%
  \footnotesize{\color{blue} \textcopyright \the\year{} IEEE. Personal use of this material is permitted. Permission from IEEE must be obtained for all other uses, including reprinting/republishing this material for advertising or promotional purposes, collecting new collected works for resale or redistribution to servers or lists, or reuse of any copyrighted component of this work in other works.}}

\newcommand\copyrightnotice{%
\begin{tikzpicture}[remember picture,overlay]
\node[anchor=north,yshift=-10pt] at (current page.north) {\fbox{\parbox{\dimexpr0.9\textwidth-\fboxsep-\fboxrule\relax}{\copyrighttext}}};
\end{tikzpicture}%
}
\renewcommand\fbox{\fcolorbox{red}{white}}
\begin{document}
\title{Hierarchical Decision-Making in Population Games}
\author{Yu-Wen Chen, 
Nuno C. Martins,
Murat Arcak
\thanks{This work was supported by NSF Grants 2135791 and 2135561 and AFOSR FA95502310467.}%
\thanks{Yu-Wen Chen and Murat Arcak are with the Department of Electrical Engineering and Computer Sciences, University of California, Berkeley, CA, USA. {\tt\small \{yuwen\_chen, arcak\}@berkeley.edu}.}%
\thanks{Nuno C. Martins is with the Department of Electrical and Computer Engineering, University of Maryland, College Park, MD, USA. {\tt\small nmartins@umd.edu}.}%
}

\maketitle

\copyrightnotice
\vspace{-10pt}

\begin{abstract}
This paper introduces a hierarchical framework for population games, where individuals delegate decision-making to proxies that act within their own strategic interests. This framework extends classical population games, where individuals are assumed to make decisions directly, to capture various real-world scenarios involving multiple decision layers. We establish equilibrium properties and provide convergence results for the proposed hierarchical structure. Additionally, based on these results, we develop a systematic approach to analyze population games with general convex constraints, without requiring individuals to have full knowledge of the constraints as in existing methods. We present a navigation application with capacity constraints as a case study.
\end{abstract}

\section{Introduction}
\label{sec:intro}

Population games study the aggregate behavior of strategic agents in large populations who adapt their actions via learning to maximize payoffs. 
A fundamental research question is determining conditions under which learning rules and payoff mechanisms ensure convergence to Nash equilibria. Early studies focused on memoryless payoff mechanisms, such as potential games \cite{Sandholm2001}. Later, building on \cite{Hofbauer2009},
the paper \cite{Fox2013} introduced the concept of $\delta$-passivity to characterize stability properties and to allow for dynamic payoff mechanisms.
A broader $\delta$-dissipativity approach was presented in \cite{Arcak2020}. Other passivity approaches were applied to imitation dynamics in \cite{Mabrok2021}.
Further advancing this line of work, \cite{Martins2024} brought the notion of counterclockwise passivity (CCW) \cite{Angeli2006} to  population games.

While the studies discussed above assume that individuals directly select their strategies, in many applications, decision-making is mediated by proxies with more information or computational power. For example, investors delegate portfolio management to financial managers, and travelers follow navigation applications in selecting routes. 
This leads to hierarchical, rather than direct,
decision-making 
as illustrated in \Cref{fig:direct,fig:toy_diagram}. 
Moreover, proxies may have 
strategy preferences
that constrain the social state to a subset of the simplex in which
it evolves.
Limited work has addressed population games with such constraints.
In \cite{Julian2018}, the authors considered box constraints which do not account for coupling restrictions between strategies. Moreover, the work requires a demanding information structure where the individuals have access to all system-wide constraints. In \cite{Juan2022}, 
a payoff penalty  was used to ensure convergence to the feasible set, but the constraints were not satisfied at all times. Recently, \cite{Chen2024} studied constrained best-response dynamics for convex constraints; however, as in \cite{Julian2018}, these constraints are assumed to be common knowledge.


The main contributions of this paper are twofold. First, we introduce a hierarchical framework for population games,
which captures practical scenarios involving multiple decision layers.
With properly designed proxies and payoff structure, we provide convergence results akin to those in classical population games. Second, we 
address population games with general convex constraints within this framework. Unlike existing methods \cite{Julian2018,Juan2022,Chen2024} requiring constraints known to individuals, we incorporate constraints through the hierarchical structure.

\Cref{sec:pre} gives an overview of population games. \Cref{sec:prob} introduces the hierarchical framework and formulates the main problem. \Cref{sec:main} presents the theoretical results. \Cref{sec:app} illustrates the results on a navigation example. 

\section{Preliminaries}
\label{sec:pre}

For a classical (single) population game, each individual in the population can choose from a set of strategies $\mathcal{S}=\{1,\ldots,d\}$. Let $\Delta^d=\left\{\bm{v}\in\mathbb{R}^d_+:\bm{1}^T\bm{v}=1\right\}$ be the probability simplex in $\mathbb{R}^d$. Denote the state $\bm{s}(t)=[s_1(t),\ldots,s_d(t)]^T\in\Delta^d$, where $s_i(t)$ describes the proportion of the population selecting strategy~$i$ at time $t$. A payoff function $F:\Delta^d\rightarrow\mathbb{R}^d$ maps a state $\bm{s}(t)$ to a payoff vector $F(\bm{s}(t))$ with $F_i(\bm{s}(t))$ representing the payoff for choosing strategy $i$.

\begin{definition}[Best response to $\bm{\pi}$ within $\mathcal{C}$]\label{def:BR}
    The best response to a payoff vector $\bm{\pi}\in\mathbb{R}^d$ within a compact set $\mathcal{C}\subseteq\Delta^d$ is denoted as $\textit{BR}_{\mathcal{C}}(\bm{\pi})=\mathop{\arg\max}_{\bm{s}\in\mathcal{C}}\ \bm{s}^T\bm{\pi}$.
\end{definition}

\begin{definition}[Nash equilibrium of $F$ within $\mathcal{C}$]\label{def:NE}
    Let $\mathcal{C}\subseteq\Delta^d$ be compact. 
    $\bm{s}^*\in\mathcal{C}$ is a Nash equilibrium for the payoff function $F$ within $\mathcal{C}$ if $(\bm{s}-\bm{s}^*)^TF(\bm{s}^*)\leq0$, for $\bm{s}\in\mathcal{C}$.
    The set of all Nash equilibria for $F$ within $\mathcal{C}$ is denoted as $\text{NE}_{\mathcal{C}}(F)$.
\end{definition}
 
Given payoff vectors $\bm{\pi}(t)\in\mathbb{R}^d$ for all $t$, the individuals switch between strategies to receive higher payoffs, which leads to a dynamics in $\bm{s}(t)$. An evolutionary dynamics model (EDM) is used to describe this learning dynamics:
\begin{align}
    \dot{\bm{s}}(t)=\mathcal{V}(\bm{s}(t),\bm{\pi}(t)), \quad t\geq0, \label{eq:EDM}
\end{align}
where $\mathcal{V}:\Delta^d\times\mathbb{R}^d\rightarrow\mathbb{R}^d$ is defined based on the switching behaviors. Some common EDMs include Smith \cite{Smith1984}, Brown-von Neumann-Nash (BNN) \cite{Brown1950}, and best response dynamics \cite[Chapter~6]{Sandholm2010}. While the best response dynamics is defined over the simplex $\Delta^d$, a generalization to it is defined over a compact convex subset of $\Delta^d$ as follows.




\begin{definition}[Constrained best response dynamics \cite{Chen2024}]
    \begin{align}
        \dot{\bm{s}}(t)\in \textit{BR}_\mathcal{C}\left(\bm{\pi}(t)\right) - \bm{s}(t), \quad\mathcal{C}\subseteq\Delta^d\text{ compact convex}.\label{eq:CBR}
    \end{align}
\end{definition}

\begin{definition}[Positive Correlation]
    An EDM (\ref{eq:EDM}) is positively correlated if $\mathcal{V}(\bm{s},\bm{\pi})\neq \bm{0} \implies \bm{\pi}^T\mathcal{V}(\bm{s},\bm{\pi})>0$.
\end{definition}

\begin{definition}[Nash Stationarity w.r.t. $\mathcal{C}$]
    An EDM (\ref{eq:EDM}) is Nash stationary w.r.t. $\mathcal{C}$ if $\mathcal{V}(\bm{s},\bm{\pi})=\bm{0}\iff \bm{s}\in \textit{BR}_\mathcal{C}(\bm{\pi})$.
\end{definition}
While BNN, Smith, and best response dynamics are Nash stationary w.r.t. $\Delta^d$, (\ref{eq:CBR}) is Nash stationary w.r.t. $\mathcal{C}$.\footnote{Since the best response dynamics and (\ref{eq:CBR}) are differential inclusions, a modified version of the Nash stationarity is used, as in \cite[Theorem~6.1.4]{Sandholm2010}.}

\section{Problem Formulation}
\label{sec:prob}

We first present the core formulation and introduce the key notation using a two-layer example in \Cref{subsec:toy_re}. We extend this to the general hierarchical framework in \Cref{subsec:hier}, and formulate several problems in this framework in \Cref{subsec:description}.
\subsection{Illustrative example}
\label{subsec:toy_re}


Consider a group of investors who allocate capital across three investment targets through two managers. This leads to two layers as in \Cref{fig:toy_diagram}: the first contains the investors; the second comprises two managers. We assume each investor has one unit of capital and must allocate it to either manager. Manager 1 distributes the funds across three investment targets, while Manager 2 invests only in the first and second.

Each decision group has a state variable representing its strategy distribution over its strategy set. Specifically, for each decision group~$j$ in layer~$i$, denoted as the $(i,j)$-group, we define a state variable $\bm{s}^{i,j}(t)\in\Delta^{d^{i,j}}$, where $d^{i,j}$ is the number of available strategies to the $(i,j)$-group. The $k$-th entry of $\bm{s}^{i,j}(t)$, denoted as $s^{i,j}_k(t)$, represents the proportion of the $(i,j)$-group selecting the $k$-th strategy at time $t$.

We define the social state $\bm{x}(t)\in\Delta^3$ as the final distribution over the three final strategies (investment targets) at time $t$:
\begin{subequations}
\begin{align}
    \bm{x}(t)&=\begin{bmatrix}
        x_1(t)\\x_2(t)\\x_3(t)
    \end{bmatrix}=\begin{bmatrix}
        s^{1,1}_1(t) s^{2,1}_1(t)+s^{1,1}_2(t)s^{2,2}_1(t) \\
        s^{1,1}_1(t) s^{2,1}_2(t)+s^{1,1}_2(t)s^{2,2}_2(t) \\
        s^{1,1}_1(t) s^{2,1}_3(t)
    \end{bmatrix}\label{eq:toy_x}
    \\
    &=\underbrace{\begin{bmatrix}
        1 & 0 & 0 & 1 & 0 \\
        0 & 1 & 0 & 0 & 1 \\
        0 & 0 & 1 & 0 & 0
    \end{bmatrix}}_{:=\bm{W}^2 \in \mathbb{R}^{3\times5}}
    \underbrace{\vphantom{\begin{bmatrix}
        0 \\ 0 \\ 0
    \end{bmatrix}}\begin{bmatrix}
        \bm{s}^{2,1}(t) & \bm{0}_{3\times1} \\
        \bm{0}_{2\times1} & \bm{s}^{2,2}(t)
    \end{bmatrix}}_{:=\bm{T}^2(t) \in \mathbb{R}^{5\times2}}
    \bm{s}^{1,1}(t).\label{eq:toy_x_compact}
\end{align}
\end{subequations}
Equation (\ref{eq:toy_x}) follows from the paths in \Cref{fig:toy_diagram}, while (\ref{eq:toy_x_compact}) provides a compact representation. Here, $\bm{T}^2(t)$ models the redistribution of funds by the managers in layer 2 at time $t$, and $\bm{W}^2$ aggregates the resulting allocations. Further details on $\bm{T}^2(t)$ and $\bm{W}^2$ are provided in \Cref{subsec:hier}.

\begin{figure}[!t]
    \centering
    \begin{subfigure}[t]{0.17\textwidth}
    \includegraphics[height=.12\textheight]{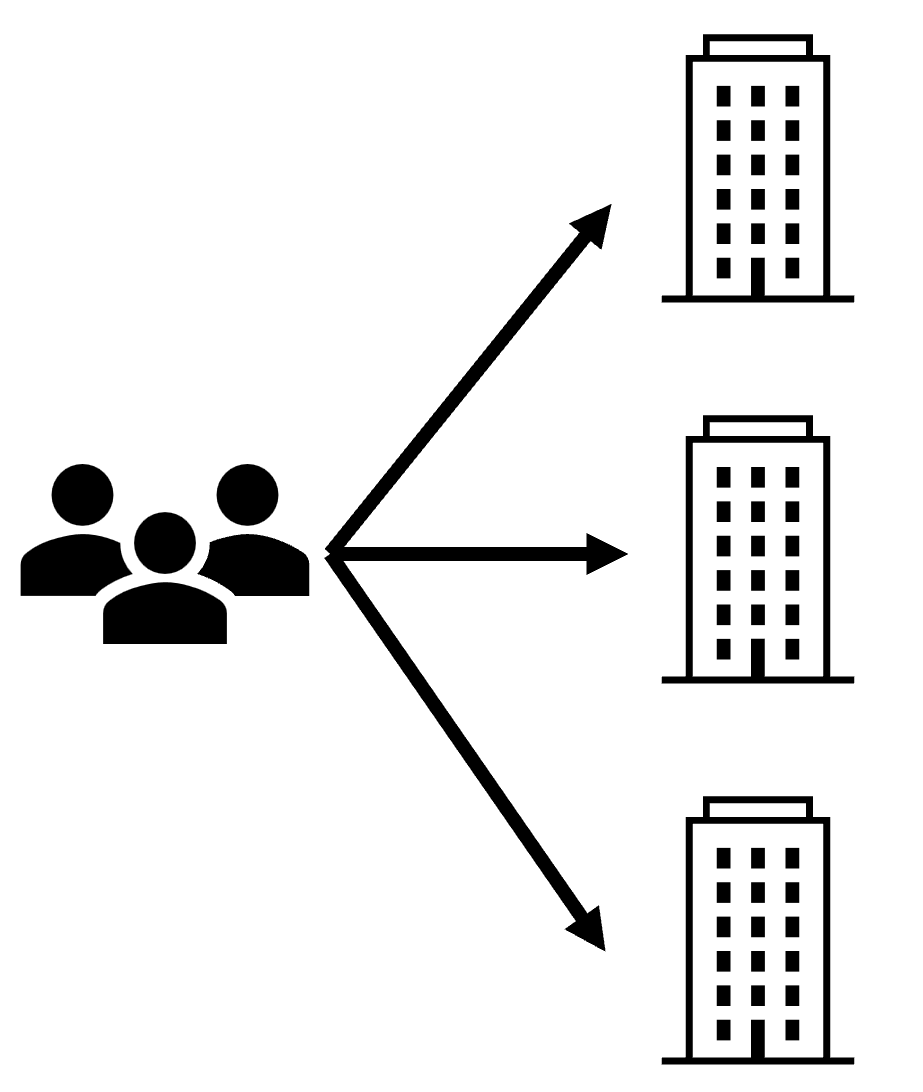}
    \caption{Direct framework}
    \label{fig:direct}
    \end{subfigure}
    \quad
    \begin{subfigure}[t]{0.27\textwidth}
    \centering
    \includegraphics[height=.14\textheight]{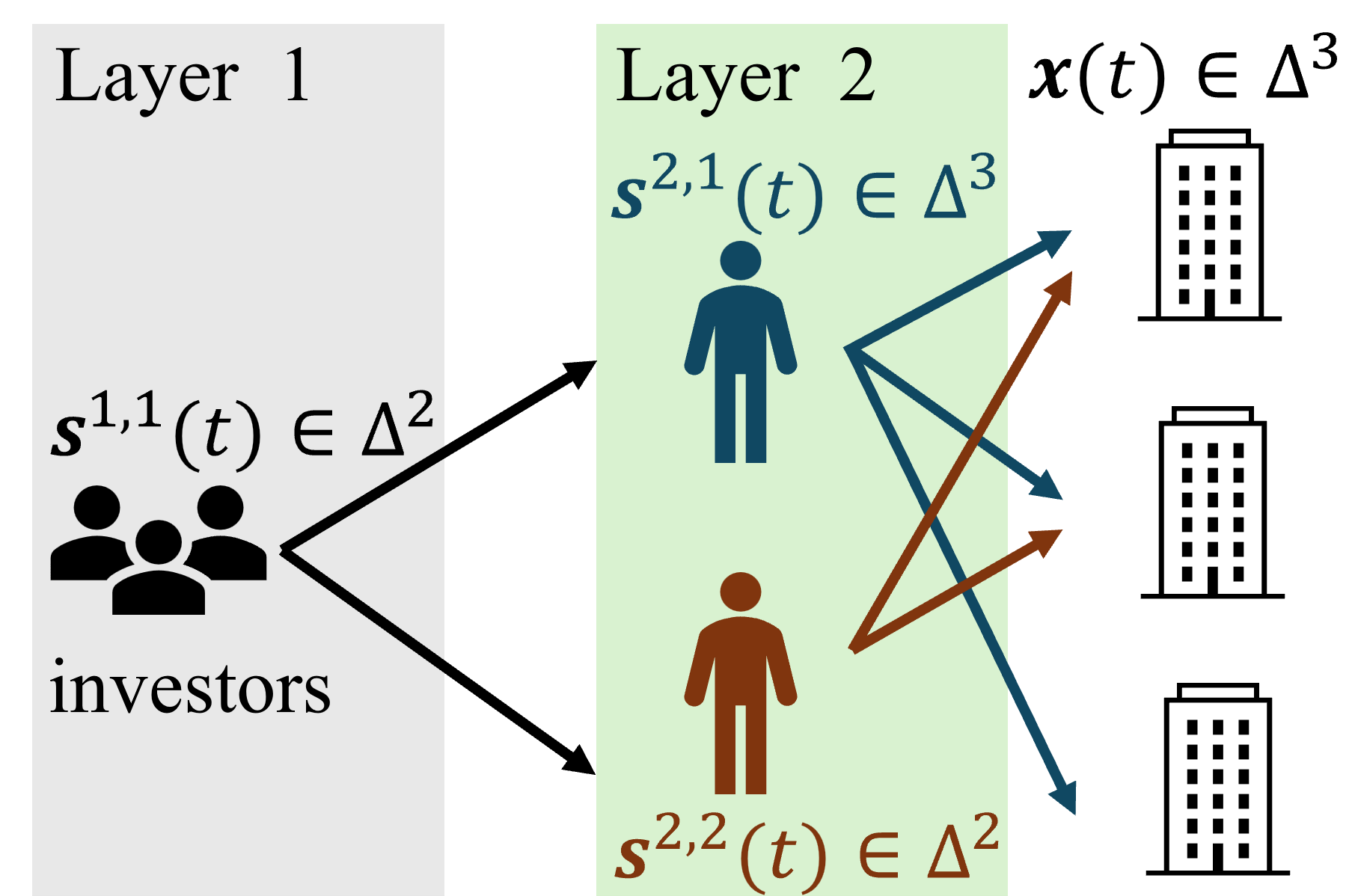}
    \caption{Two-layer hierarchy}
    \label{fig:toy_diagram}
    \end{subfigure}
    \par\medskip
    \begin{subfigure}[t]{0.48\textwidth}
    \centering
    \includegraphics[width=\textwidth]{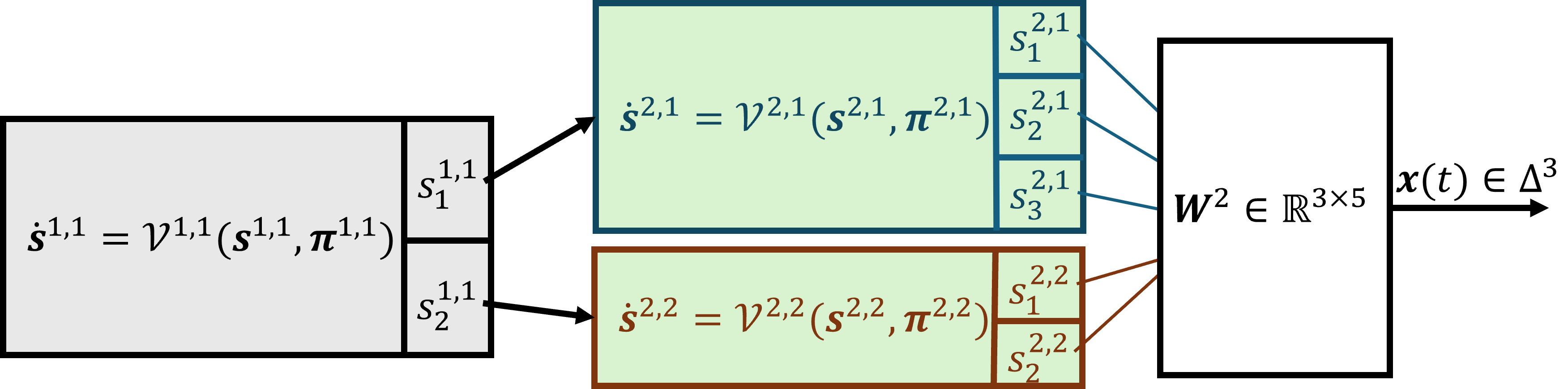}
    \caption{Two-layer population game formulation for (b)}
    \label{fig:toy_math}
    \end{subfigure}

    \setlength{\belowcaptionskip}{-13pt}
    \caption{In contrast to direct framework (a) in most literature, we provide an illustrative example with two-layer hierarchy.}
    \label{fig:toy}
    \setlength{\belowcaptionskip}{0pt}
\end{figure}

With the defined state variables $\bm{s}^{i,j}(t)$ and social state $\bm{x}(t)$, we now describe the system dynamics. A payoff function $F:\Delta^3\rightarrow\mathbb{R}^3$ maps the social state $\bm{x}(t)$ to a payoff vector $F(\bm{x}(t))=\left[F_1(\bm{x}(t)),F_2(\bm{x}(t)),F_3(\bm{x}(t))\right]^T$, where $F_i(\bm{x}(t))$ represents the payoff for investing in the $i$-th investment target. Manager 1 updates portfolio allocations based on $F(\bm{x}(t))$, while Manager 2 updates based on $F_1(\bm{x}(t))$ and $F_2(\bm{x}(t))$. In addition, investors update their choices based on manager performance. A natural performance metric for managers is their average payoff. Therefore, we can model the updates by:
\begin{align}
    \dot{\bm{s}}^{i,j}(t)&=\mathcal{V}^{i,j}\left(\bm{s}^{i,j}(t),\bm{\pi}^{i,j}(t)\right) \label{eq:toy_dyn},\quad t\geq0,
\end{align}
for $(i,j)\in\left\{(1,1),(2,1),(2,2)\right\}$. Here, $\mathcal{V}^{i,j}$ is the EDM defined in \Cref{sec:pre} and $\bm{\pi}^{i,j}(t)$ is the payoff at time $t$, for the $(i,j)$-group, while the payoffs are given as:
\begin{subequations}\label{eq:toy_pay}
\begin{align}
   \bm{\pi}^{2,1}(t)&=F(\bm{x}(t)) \label{eq:toy_pay_first} \\
   \bm{\pi}^{2,2}(t)&=\left[F_1(\bm{x}(t)),F_2(\bm{x}(t))\right]^T\label{eq:toy_pay_2nd}\\
   \bm{\pi}^{1,1}(t)&=\scalebox{1}{$ \left[\sum\limits_{l=1}^3s^{2,1}_l(t)F_l(\bm{x}(t)),\sum\limits_{l=1}^2s^{2,2}_l(t)F_l(\bm{x}(t))\right]^T$.}\label{eq:toy_pay_last}
   \end{align}
\end{subequations}
The overall structure is visualized in \Cref{fig:toy_diagram,fig:toy_math}.

\subsection{Hierarchical framework}
\label{subsec:hier}

Generalizing the notation $\bm{T}^i(t)$ and $\bm{W}^i$ introduced in (\ref{eq:toy_x_compact}), we define a transformation matrix $\bm{T}^i(t)$ by diagonally concatenating all state variables $\bm{s}^{i,j}(t)\in\Delta^{d^{i,j}}$ in layer~$i$. Suppose there are $n^i$ decision groups (state variables) in layer~$i$, then for all $t\geq0$,
\begin{align}
    \bm{T}^i(t)= \mathrm{blkdiag}\left(\bm{s}^{i,1}(t),\ldots,\bm{s}^{i,n^i}(t)\right) \in\mathbb{R}^{o^i \times n^i},\label{eq:T}
\end{align}
where $\mathrm{blkdiag}$ denotes the block diagonal matrix composed by the arguments and $o^i=\sum_j d^{i,j}$ is the total number of outputs of layer~$i$.
Since outputs from one layer may be grouped into decision groups in the next layer, we introduce an aggregation matrix to formalize this transition. These aggregation matrices define how outputs from layer~$i$ are assigned to decision groups in layer~$i+1$, fully characterizing the hierarchical structure. An aggregation matrix is defined as:
\begin{align}
    \bm{W}^i=\begin{bmatrix}
        \bm{w}_1^i,\ldots,\bm{w}^i_{o^i}
    \end{bmatrix}\in\mathbb{R}^{n^{i+1}\times o^i},\label{eq:W}
\end{align}
where $\bm{w}_k^i\in\mathbb{R}^{n^{i+1}}$ represents how the $k$-th output of layer~$i$ distributes to the next layer.
Let $n^{L+1}=d$ represent the number of final strategies.
The transformation matrix $\bm{T}^i(t)$ and aggregation matrix $\bm{W}^i$ satisfy the following properties:
\begin{property}\label{p:nonnegative}
    $\bm{T}^i(t)$ and $\bm{W}^i$ are elementwise nonnegative.
\end{property}
\begin{property}\label{p:sum_one}
    $\bm{1}^T\bm{T}^i(t)=\bm{1}^T$ and $\bm{1}^T\bm{W}^i=\bm{1}^T$.
\end{property}

We now formulate the general $L$-layer hierarchical framework using the notation introduced in (\ref{eq:T}) and (\ref{eq:W}). Suppose layer 1 consists of a single decision group representing the entire population, while each subsequent layer $i$ has $n^i \geq 1$ decision groups for $i = 2, \dots, L$. The decision group $j$ in layer $i$ is associated with a state variable $\bm{s}^{i,j}(t) \in \Delta^{d^{i,j}}$, which describes the group’s distribution over the group's strategy set.

Let the hierarchy be characterized by $\bm{W}^1$ to $\bm{W}^L$.
Note that since layer $1$ has only one decision group, $\bm{T}^1(t)=\bm{s}^{1,1}(t)$ and $\bm{W}^1$ is the identity matrix.
Then, we can derive the distribution of the entire population input to layer $i$ at time $t$ by
\begin{align}
    \bm{m}^{i}(t)=\bm{W}^{i-1}\bm{T}^{i-1}(t)\ldots\bm{W}^1\bm{s}^{1,1}(t)\in\Delta^{n^i}.\label{eq:M}
\end{align}
Further, $\bm{W}^L\bm{T}^L(t)\bm{m}^L(t)$ captures the distribution of the entire population over the $d$ final strategies, which we call the social state. We denote the social state as $\bm{x}(t)\in\Delta^d$:
\begin{align}
    \bm{x}(t)&=\bm{W}^L\bm{T}^L(t)\bm{m}^L(t)=\Pi_{i=1}^L\bm{W}^i\bm{T}^i(t) \nonumber\\
    &=\bm{W}^L\bm{T}^L(t)\ldots\bm{W}^2\bm{T}^2(t)\bm{W}^1\bm{s}^{1,1}(t), \quad t\geq0. \label{eq:x}
\end{align}
Using \Cref{p:nonnegative} and \Cref{p:sum_one}, we can verify that $\bm{x}(t)\in\Delta^d$, confirming that the population is properly distributed over the final $d$ strategies. However, unlike in \Cref{fig:direct}, where the individual selects a final strategy directly, here, decisions are made indirectly through multiple layers. Therefore, we refer to each layer $i$, for $i=2,\dots,L$, as a proxy layer.

To complete the formulation, we describe the dynamics for $\bm{s}^{i,j}(t)$. Let $\bm{p}(t)\in\mathbb{R}^d$ represent the $d$ payoffs for selecting each final strategy. We temporarily assume $\bm{p}(t)$ is given by a payoff function $F: \Delta^d\rightarrow\mathbb{R}^d$, mapping the social state $\bm{x}(t)$ to the payoff vector $F(\bm{x}(t))$ for all $t$:
\begin{align}
    \bm{p}(t)=F(\bm{x}(t)), \quad t\geq0. \label{eq:F}
\end{align}
Later, we will extend this to a dynamic model in \Cref{subsec:dyn}.
Denote $\bm{\pi}^{i,j}(t) \in \mathbb{R}^{d^{i,j}}$ as the $d^{i,j}$ payoffs for selecting each strategy in $(i,j)$-group. Recall from \Cref{subsec:toy_re} that $\bm{\pi}^{i,j}(t)$ is determined by the average payoff of the group it joins via back-propagating $\bm{p}(t)$, as in (\ref{eq:toy_pay}). 
Let $\bm{\pi}^i(t)$ be the concatenation of payoffs $\bm{\pi}^{i,j}(t)$ in layer $i$:
\begin{align}
    \bm{\pi}^i(t)=\begin{bmatrix}
        \bm{\pi}^{i,1}(t) \\ \vdots \\ \bm{\pi}^{i,n^i}(t)
    \end{bmatrix}\in\mathbb{R}^{o^i}, \quad t\geq0. \label{eq:pi_stack}
\end{align}
Then, the payoffs are determined by
\begin{subequations}\label{eq:pi}
\begin{align}
    \bm{\pi}^L(t)&={\bm{W}^L}^T\bm{p}(t)\label{eq:pi_L}\\
    \bm{\pi}^i(t)&={\bm{W}^i}^T{\bm{T}^{i+1}(t)}^T\bm{\pi}^{i+1}(t), \quad i=L-1,\ldots,1.\label{eq:pi_i}
\end{align}
\end{subequations}
With (\ref{eq:pi}), $\bm{s}^{i,j}(t)$ evolves and follows, as in (\ref{eq:toy_dyn}), the EDM:
\begin{align}
    \dot{\bm{s}}^{i,j}(t)=\mathcal{V}^{i,j}\left(\bm{s}^{i,j}(t),\bm{\pi}^{i,j}(t)\right), \quad t\geq0. \label{eq:dyn}
\end{align}
The hierarchical framework is illustrated in \Cref{fig:hier_diagram}.

\begin{figure}[!ht]
    \centering
    \includegraphics[width=\linewidth]{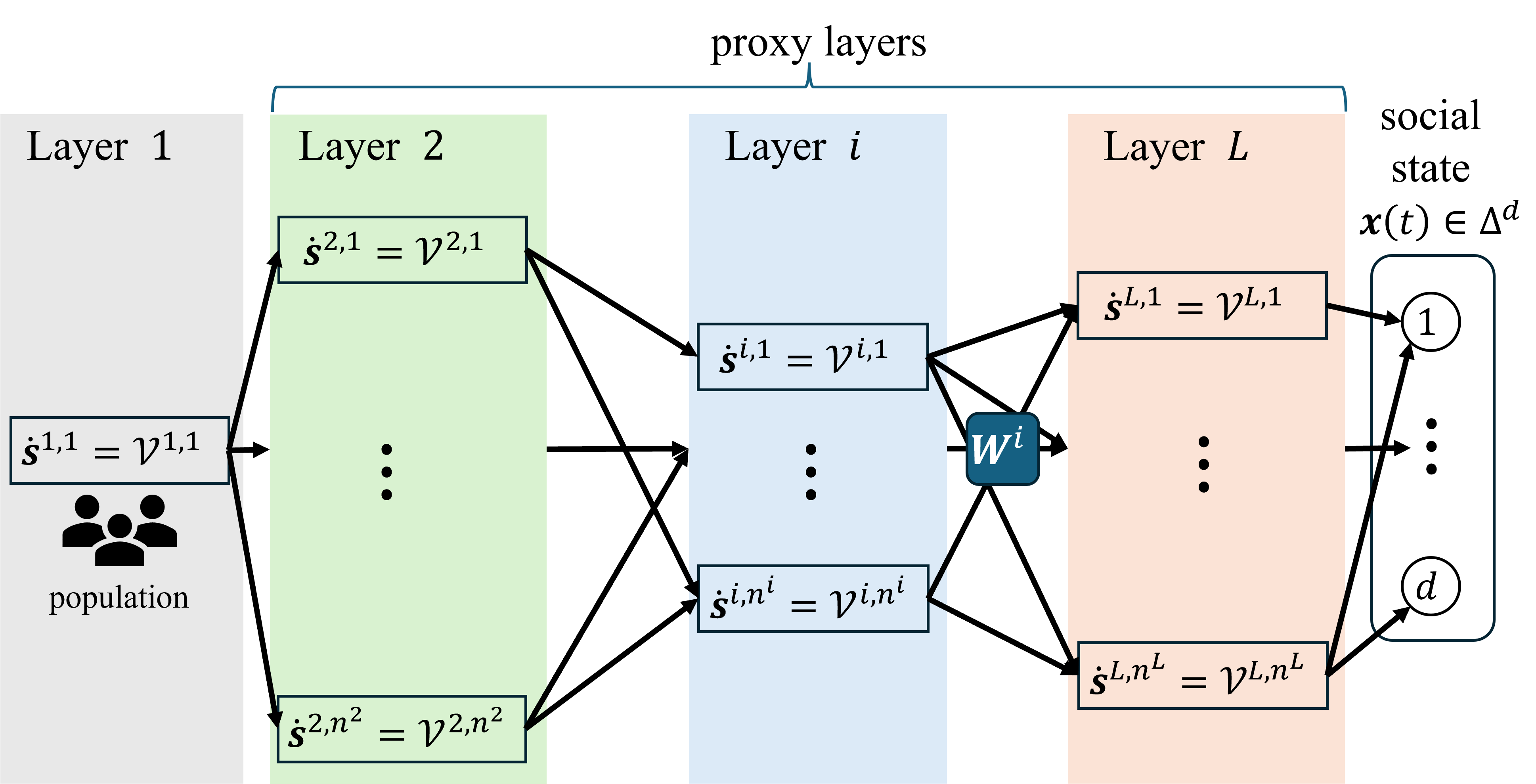}
    \caption{Illustration of the hierarchical framework. The population contributes to the social state via multiple proxy layers.}
    \label{fig:hier_diagram}
\end{figure}

\subsection{Problem description}
\label{subsec:description}

In this paper, we consider the general case where the \emph{strategy distribution of interest} for each $(i,j)$-decision group is not necessarily the entire simplex but a subset $\mathcal{K}^{i,j}\subseteq\Delta^{d^{i,j}}$. This states that the $(i,j)$-group always makes a decision within $\mathcal{K}^{i,j}$, i.e., $\bm{s}^{i,j}(t)\in \mathcal{K}^{i,j}$. These subsets collectively determine an \emph{admissible set} $\mathcal{K}\subseteq\Delta^d$, where the social state $\bm{x}(t)$ can reside. Denote $[m]$ as the set $\{1,\ldots,m\}$.
\begin{definition}[Admissible set $\mathcal{K}$]\label{def:admissible}
\begin{align}
    \mathcal{K}&=\left\{\Pi_{i=1}^L\bm{W}^i\bm{K}^i: \bm{k}^{i,j}\in\mathcal{K}^{i,j},i\in [L], j\in [n^i]\right\}, \label{eq:K}
\end{align}
where
$\bm{K}^i=\mathrm{blkdiag}\left(\bm{k}^{i,1},\ldots,\bm{k}^{i,n^i}\right)\in\mathbb{R}^{o^i\times n^i}$.
\end{definition}
The admissible set $\mathcal{K}$ collects all possible social states generated from different combinations of group decisions. Consequently, (\ref{eq:K}) follows the same structures as (\ref{eq:x}) and (\ref{eq:T}). We provide a concrete example.

\begin{example}[Strategy distribution of interest]
\label{ex:K}
    Consider the example in \Cref{subsec:toy_re}. Suppose that Manager~1 favors the distribution $\bm{s}^{2,1}_*=[0.2,0.2,0.6]^T$ over the three investment targets but allows variations up to $\varepsilon=0.1$, while Manager~2 prefers to invest at least $1/3$ in both investment targets. Then, the strategy distributions of interest are
\begin{subequations}
\begin{align}
    \mathcal{K}^{2,1}&=\left\{\bm{x}\in\Delta^3:\left\|\bm{x}-\bm{s}^{2,1}_*\right\|_2\leq\varepsilon\right\}\quad\text{and}\\
    \mathcal{K}^{2,2}&=\left\{\bm{x}\in\Delta^2:x_1\geq\frac{1}{3}, x_2\geq\frac{1}{3}\right\}.
\end{align}
\end{subequations}
\begin{minipage}{\columnwidth}
\begin{wrapfigure}{R}{.38\columnwidth}
    \centering\includegraphics[width=.37\columnwidth]{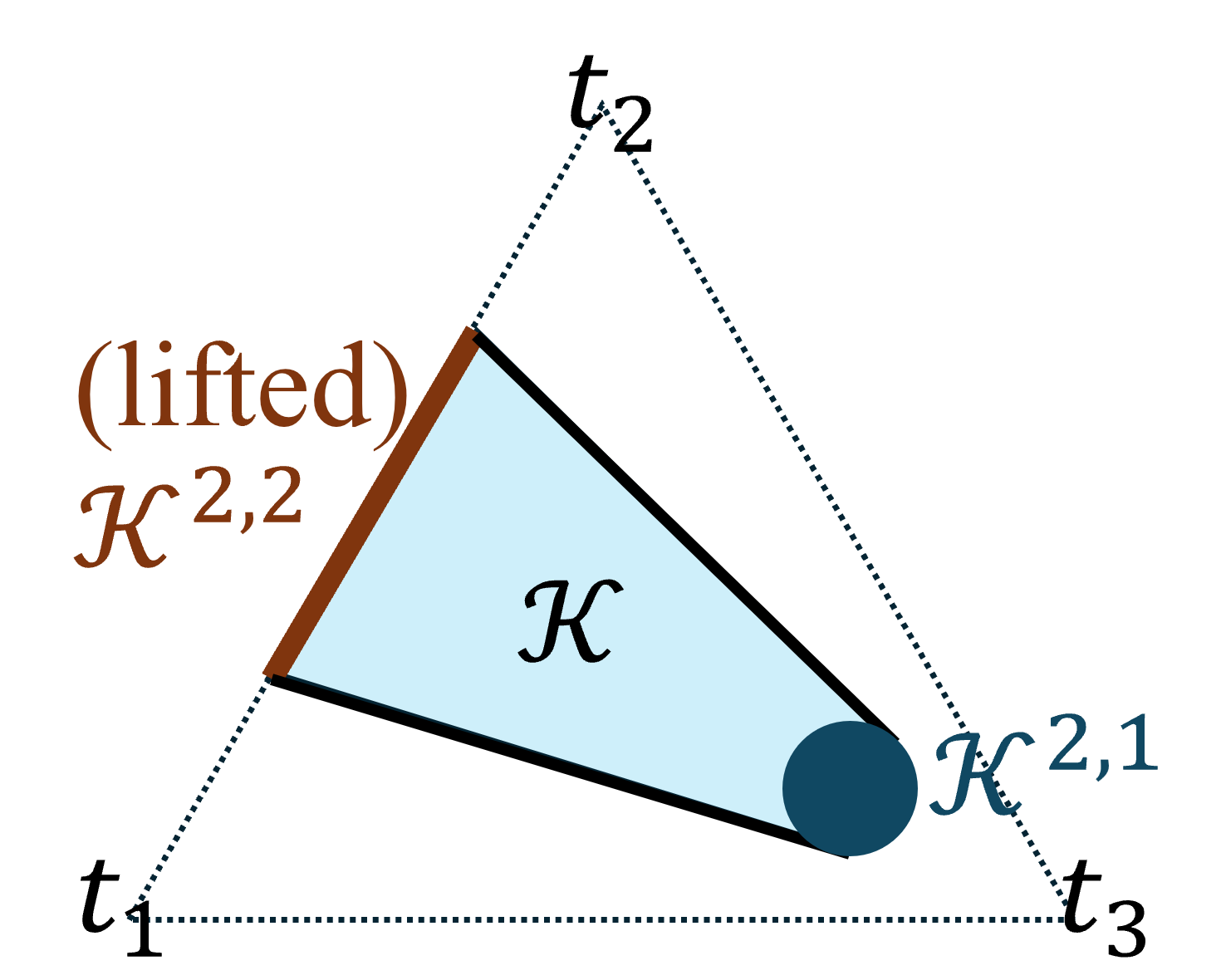}
    \caption{Strategy distributions of interest over 3 investment targets.}
    \label{fig:sets}
\end{wrapfigure}
Elements of $\mathcal{K}^{2,2}$ are two-dimensional, corresponding to the first two final strategies. To visualize $\mathcal{K}^{2,2}$ together with $\mathcal{K}^{2,1}$ in $\Delta^3$, we lift elements in $\mathcal{K}^{2,2}$ to 3D by appending a zero. If investors are free to choose between the two managers, i.e., $\mathcal{K}^{1,1}=\Delta^2$, then the admissible set $\mathcal{K}$ is the convex hull of $\mathcal{K}^{2,1}$ and the lifted $\mathcal{K}^{2,2}$, as shown in \Cref{fig:sets}.
\end{minipage}
\end{example}


\textbf{Main Problem}: Given the hierarchical framework (\ref{eq:x})-(\ref{eq:dyn}), we aim to determine conditions such that the social state $\bm{x}(t)$ converges to a Nash equilibrium characterized by $F$ within the admissible set $\mathcal{K}$ (\ref{eq:K}).
This framework introduces new complexities due to multiple layers and additional challenges from the coupling dynamics among decision groups and their collective influence on $\bm{x}(t)$ via (\ref{eq:x}).

As will be illustrated in  \Cref{sec:app}, our solution suggests a novel approach to constraining the social state within a desired set $\mathcal{D}\subseteq\Delta^d$, e.g. to avoid socially inefficient states. If we can design strategy distributions of interest $\mathcal{K}^{i,j}$s such that the resulting admissible set $\mathcal{K}$ satisfies $\mathcal{K}\subseteq\mathcal{D}$, and the dynamics meets our conditions, then $\bm{x}(t)$ evolves and converges within $\mathcal{D}$. Moreover, if we design $\mathcal{K}^{1,1}=\Delta^{d^{1,1}}$, then this approach does not impose direct constraints on individual decisions; instead, the constraint is embedded into the strategy distributions of the proxies.

We impose mild assumptions derived from the classical population games. Since not all results require all assumptions, we will state explicitly which assumptions apply in each case.

\begin{assumption}\label{ass:convex}
    $\mathcal{K}^{i,j},i\in[L],j\in[n^i]$, are compactly convex.
\end{assumption}

\begin{assumption}\label{ass:F}
    $F$ is continuously differentiable.    
\end{assumption}

\begin{assumption}\label{ass:dyn}
    The learning dynamics (\ref{eq:dyn}) for each $(i,j)$-group ensures that $\mathcal{K}^{i,j}$ is forward invariant. 
\end{assumption}

\begin{assumption}\label{ass:conti}
    $\bm{\pi}^T\mathcal{V}^{i,j}(\bm{s},\bm{\pi})$, $i\in[L],j\in[n^i]$, are Lipschitz continuous w.r.t. $\bm{s}$ and $\bm{\pi}$.
\end{assumption}

When $L=1$ and $\mathcal{K}^{1,1}=\Delta^{d^{1,1}}$, the assumptions reduce to the ones in classical population games.
Here, \Cref{ass:dyn} means that each $(i,j)$-group makes decisions consistently within its strategy distribution of interest $\mathcal{K}^{i,j}$.
Generally, if $\mathcal{K}^{i,j}$ is a convex set, the constrained best response dynamics (\ref{eq:CBR}) guarantees the forward invariance. Moreover, all learning dynamics introduced in \Cref{sec:pre} satisfy \Cref{ass:conti}.

\section{Analysis}
\label{sec:main}
We focus on static models for payoff (\ref{eq:F}) in \Cref{subsec:static} and extend to dynamic models in \Cref{subsec:dyn}.
\subsection{Static Models for Payoff}
\label{subsec:static}
We first characterize the system at its rest points, showing that the corresponding social state is a Nash equilibrium of $F$ over the admissible set $\mathcal{K}$. We then establish conditions under which the system converges to these rest points.

\begin{theorem}[Characterization at rest points] \label{thm:NE}
    Suppose each $\mathcal{V}^{i,j}$ in (\ref{eq:dyn}) is Nash stationary w.r.t. $\mathcal{K}^{i,j}$. Let $\bar{\bm{x}}\in\mathcal{K}$ be an equilibrium social state when the system is at rest. Then, $\bar{\bm{x}}\in \text{NE}_{\mathcal{K}}(F)$.
\end{theorem}

\begin{proof}
    We use a bar on a variable to denote its value at a rest point of (\ref{eq:dyn}). Since $\mathcal{V}^{i,j}$s are Nash stationary w.r.t. $\mathcal{K}^{i.j}$s, by definition we have $\left(\bm{s}^{i,j}-\bar{\bm{s}}^{i,j}\right)^T\bar{\bm{\pi}}^{i,j}\leq0$, for $\bm{s}^{i,j}\in\mathcal{K}^{i,j}$.
    Using (\ref{eq:T}), (\ref{eq:pi_stack}), and (\ref{eq:K}), we rewrite it compactly as
    \begin{align}
    \left(\bm{K}^i-\bar{\bm{T}}^i\right)^T\bar{\bm{\pi}}^i \preceq \bm{0}, \quad \bm{k}^{i,j}\in\mathcal{K}^{i,j}, i\in[L], j\in[n^i], \label{eq:rec}
    \end{align}
    where $\preceq$ denotes elementwise comparisons. For $i=1$, since $\bm{K}^1=\bm{k}^{1,1}$ is a vector, we have
    \begin{subequations}
    \begin{align}
        &\left(\bm{K}^1-\bar{\bm{T}}^1\right)^T\bar{\bm{\pi}}^1\leq0\\
        \iff&\left(\bar{\bm{T}}^2\bm{W}^1\bm{K}^1-\bar{\bm{T}}^2\bm{W}^1\bar{\bm{T}}^1\right)^T\bar{\bm{\pi}}^2\leq0,\label{eq:to_add_1}
    \end{align}
    \end{subequations}
    where (\ref{eq:to_add_1}) follows from (\ref{eq:pi_i}). For $i=2$ in (\ref{eq:rec}), we obtain $\left(\bm{K}^2-\bar{\bm{T}}^2\right)^T\bar{\bm{\pi}}^2\preceq\bm{0}$.
    By \Cref{p:nonnegative}, $\bm{W}^1\bm{K}^1\succeq\bm{0}$ and then
    \begin{subequations}
    \begin{align}
        &\left(\bm{W}^1\bm{K}^1\right)^T\left(\bm{K}^2-\bar{\bm{T}}^2\right)^T\bar{\bm{\pi}}^2\leq0\\
        \iff&\left(\bm{K}^2\bm{W}^1\bm{K}^1-\bar{\bm{T}}^2\bm{W}^1\bm{K}^1\right)^T\bar{\bm{\pi}}^2\leq0. \label{eq:to_add_2}
    \end{align}
    \end{subequations}
    From (\ref{eq:to_add_1}) and (\ref{eq:to_add_2}), we get, for $\bm{k}^{1,j}\in\mathcal{K}^{1,j}$, $\bm{k}^{2,j}\in\mathcal{K}^{2,j}$,
    \begin{align}
        \left(\bm{K}^2\bm{W}^1\bm{K}^1-\bar{\bm{T}}^2\bm{W}^1\bar{\bm{T}}^1\right)^T\bar{\bm{\pi}}^2\leq0.
    \end{align}

    By recursively expanding $\bar{\bm{\pi}}^i$ using (\ref{eq:pi_i}) and combining with $(\bm{W}^i\bm{K}^i\ldots\bm{W}^1\bm{K}^1)^T(\bm{K}^{i+1}-\bar{\bm{T}}^{i+1})^T\bar{\bm{\pi}}^{i+1}\leq0$, we get
    \begin{align}
        &\Big(\bm{K}^L\bm{W}^{L-1}\bm{K}^{L-1}\ldots\bm{W}^1\bm{K}^1
        \nonumber\\
        &\hspace{20pt}-\bar{\bm{T}}^L\bm{W}^{L-1}\bar{\bm{T}}^{L-1}\ldots\bm{W}^1\bar{\bm{T}}^1\Big)^T{\bm{W}^L}^T\bar{\bm{p}}\leq0,
    \end{align}
    for $\bm{k}^{i,j}\in\mathcal{K}^{i,j}$. By (\ref{eq:x}), (\ref{eq:pi_L}), and (\ref{eq:K}), this simplifies to
    \begin{align}
        (\bm{x}-\bar{\bm{x}})^T\bar{\bm{p}}\leq0,\quad \bm{x}\in\mathcal{K}. \label{eq:hier_NS}
    \end{align}
    Then, the theorem follows because $\bar{\bm{p}}=F(\bar{\bm{x}})$ from (\ref{eq:F}).
\end{proof}



\Cref{thm:NE} does not require \Cref{ass:convex} to hold. However, to show convergence results in the following, convexity is crucial for learning dynamics, e.g., (\ref{eq:CBR}), to satisfy \Cref{ass:dyn}.

\begin{lemma}[Set equivalence, {\cite[Theorem~3.2]{Rockafellar1970}}]\label{lemma:set}
    If $\mathcal{C}$ is a convex set and $\alpha\geq0,\beta\geq0$, then $\alpha \mathcal{C}+\beta \mathcal{C}=(\alpha+\beta)\mathcal{C}$.
\end{lemma}

\begin{proposition}[Convexity of $\mathcal{K}$]
    If \Cref{ass:convex} holds, then $\mathcal{K}$ is convex.
\end{proposition}

\begin{proof}
    Given two distinct points $\bm{x}_p\in\mathcal{K}$, for $p=1,2$, there exists $\bm{K}^i_p$, for $i\in[L]$, such that $\bm{x}_p=\bm{W}^L\bm{K}^L_p\ldots\bm{W}^1\bm{K}^1_p\in\mathcal{K}$.
    Note that $\bm{K}^i_1$ and $\bm{K}^i_2$ belong to the same convex set, denoted by $\mathcal{K}^i$. Then, for each $\alpha,\beta\geq0,\alpha+\beta=1$, we have $\alpha\bm{x}_1+\beta\bm{x}_2$
    \begin{subequations}
    \begin{align}
        &=(\alpha\bm{W}^L\bm{K}^L_1\ldots\bm{W}^1)\bm{K}^1_1+(\beta\bm{W}^L\bm{K}^L_2\ldots\bm{W}^1)\bm{K}^1_2\label{eq:lemma}\\
        &=(\alpha\bm{W}^L\bm{K}^L_1\ldots\bm{W}^1+\beta\bm{W}^L\bm{K}^L_2\ldots\bm{W}^1)\tilde{\bm{K}}^1\\
        &=(\alpha\bm{W}^L\bm{K}^L_1\ldots\bm{K}^2_1+\beta\bm{W}^L\bm{K}^L_2\ldots\bm{K}^2_2)\bm{W}^1\tilde{\bm{K}}^1\label{eq:lemma_rec}\\
        &=\bm{W}^L\tilde{\bm{K}}^L\ldots\bm{W}^1\tilde{\bm{K}}^1\in\mathcal{K},
    \end{align}
    \end{subequations}
    where $\tilde{\bm{K}}^i\in\mathcal{K}^i$, which appears since we apply \Cref{lemma:set} to (\ref{eq:lemma}) and recursively to (\ref{eq:lemma_rec}).
\end{proof}

\Cref{thm:NE} states that 
the rest points lie in $\textit{NE}_{\mathcal{K}}(F)$. 
In the following, we give conditions for convergence to these points.

\begin{definition}[Potential game]
    A payoff function $F:\Delta^d\rightarrow\mathbb{R}^d$ is a potential game if there exists a $C^1$ potential function $f:\mathbb{R}^d\rightarrow \mathbb{R}$, such that $\nabla_{\bm{x}} f(\bm{x})=F(\bm{x})$, for $\bm{x}\in\Delta^d$.
\end{definition}

While \Cref{thm:NE} describes the analog of the Nash stationarity property for the hierarchical structure, the following lemma states the analog of the positive correlation property and is used to prove the convergence results in \Cref{thm:potential}.

\begin{lemma}\label{lemma:pc}
    If each $\mathcal{V}^{i,j}$ is positively correlated, then $\bm{p}(t)^T\dot{\bm{x}}(t)\geq0$ for all $t$.
\end{lemma}
\begin{proof}
    \begin{subequations} \label{eq:hier_pc}
    \begin{align}
        \bm{p}&(t)^T\dot{\bm{x}}(t)=\bm{p}(t)^T\Big(\bm{W}^L\dot{\bm{T}}^L(t)\bm{m}^L(t)+\ldots\nonumber\\
        &\hspace{70pt}+\bm{W}^L\bm{T}^L(t)\ldots\bm{W}^1\dot{\bm{T}}^1(t)\Big) \label{eq:by_potential}\\
        &={\bm{\pi}^L(t)}^T\dot{\bm{T}}^L(t)\bm{m}^L(t)+\ldots+{\bm{\pi}^1(t)}^T\dot{\bm{T}}^1(t)\label{eq:by_pi}\\
        &=\sum_{i=1}^L\sum_{j=1}^{n^i} m_j^i(t)\mathcal{V}^{i,j}\left(\bm{s}^{i,j}(t),\bm{\pi}^{i,j}(t)\right)^T\bm{\pi}^{i,j}(t)\geq 0.\label{eq:for_LaSalle}
    \end{align}
    \end{subequations}
    Here, (\ref{eq:by_pi}) follows from (\ref{eq:pi}), while the inequality in (\ref{eq:for_LaSalle}) holds 
    by positive correlation and \Cref{p:nonnegative}.
\end{proof}

\begin{theorem}[Convergence for potential game $F$]\label{thm:potential}
    Let $F$ be a potential game, and let \Cref{ass:convex,ass:F,ass:dyn,ass:conti} hold. Suppose that each $\mathcal{V}^{i,j}$ in (\ref{eq:dyn}) is Nash stationary w.r.t. $\mathcal{K}^{i,j}$ and positively correlated. Then, 
    $\bm{x}(t)$ asymptotically approaches $\text{NE}_{\mathcal{K}}(F)$.
\end{theorem}

\begin{proof}
    Since $F$ is a potential game, there exists a potential function $f:\mathbb{R}^d\rightarrow\mathbb{R}$. Define $f^*=\max_{\bm{x}\in\mathcal{K}}f(\bm{x})$, which exists since $\mathcal{K}$ is compact. We introduce the Lyapunov function $V:\mathcal{K}\rightarrow\mathbb{R_+}$, given by $V(\bm{x})=f^*-f(\bm{x})$. By \Cref{lemma:pc},
    \begin{align}
        \frac{d}{dt}&V(\bm{x}(t))=-\nabla_{\bm{x}}f(\bm{x}(t))^T\dot{\bm{x}}(t)=-\bm{p}(t)^T\dot{\bm{x}}(t)\leq0. \label{eq:vdot}
    \end{align}
    By (\ref{eq:for_LaSalle}) and \Cref{ass:conti}, we have $\bm{p}(t)^T\dot{\bm{x}}(t)\rightarrow0$.
    Then, by continuity, we have either $\mathcal{V}^{i,j}\left(\bm{s}^{i,j}(t),\bm{\pi}^{i,j}(t)\right)^T\bm{\pi}^{i,j}(t)\rightarrow\bm{0}$ or $m_j^i(t)\rightarrow0$. For the former, by positive correlation, we have $\bm{s}^{i,j}(t)\rightarrow\textit{BR}_{\mathcal{K}^{i,j}}\left(\bm{\pi}^{i,j}(t)\right)$. For the latter, it means no one selects the $(i,j)$-group, rendering $\bm{x}(t)$ unaffected by $\bm{s}^{i,j}(t)$. Therefore, by \Cref{thm:NE}, we conclude $\bm{x}(t)\rightarrow \textit{NE}_{\mathcal{K}}(F)$.
\end{proof}


\subsection{Dynamic Models for Payoff}
\label{subsec:dyn}
Next, we consider scenarios where the payoff, instead of being given by (\ref{eq:F}), follows the dynamical model:
\begin{align}
    \dot{\bm{q}}(t)=g(\bm{q}(t),\bm{x}(t)),\quad \bm{p}(t)=h(\bm{q}(t),\bm{x}(t)),\quad t\geq0, \label{eq:PDM}
\end{align}
where $\bm{q}(t)\in\mathbb{R}^a$ and the functions $g:\mathbb{R}^a\times\Delta^d\rightarrow\mathbb{R}^a$ and $h:\mathbb{R}^a\times\Delta^d\rightarrow\mathbb{R}^d$ are Lipschitz continuous.



\begin{definition}[Payoff Dynamics Model, PDM \cite{Park2019CDC}]
A dynamic model (\ref{eq:PDM}) is called a Payoff Dynamics Model (PDM) if it is bounded-input bounded-output stable and recovers the static model (\ref{eq:F}) in steady state, i.e.,
\begin{align}
    \lim_{t\rightarrow\infty}\|\dot{\bm{x}}(t)\|=0 \implies \lim_{t\rightarrow\infty} \bm{p}(t)=F(\bm{x}(t)). \label{eq:recover}
\end{align}
\end{definition}

Motivated by \cite{Martins2024}, we extend our results from \Cref{thm:potential} to support PDMs that satisfy the following input-output property.

\begin{definition}[Counterclockwise dissipativity, CCW \cite{Angeli2006}]
    A dynamic model (\ref{eq:PDM}) is Counterclockwise dissipative (CCW) if its input-output relation satisfies
    \begin{align}
        \alpha_c:=\liminf_{T\rightarrow+\infty} \int_0^T \dot{\bm{p}}(t)^T\bm{x}(t)\ dt > -\infty. \label{eq:CCW}
    \end{align}
\end{definition}

A potential game $F$ is equivalent to a \emph{memoryless} CCW PDM \cite{Martins2024}. Thus, the following theorem generalizes \Cref{thm:potential} in the sense that the convergence results hold for any CCW PDM. More examples of CCW dynamics can be found in \cite{Martins2024}.

\begin{theorem}[Convergence for CCW PDM]\label{thm:PDM}
    Let the payoff $\bm{p}(t)$ be given by a CCW PDM, and let \Cref{ass:convex,ass:F,ass:dyn,ass:conti} hold. Suppose that each $\mathcal{V}^{i,j}$ in (\ref{eq:dyn}) is Nash stationary w.r.t. $\mathcal{K}^{i,j}$ and positively correlated. Then, $\bm{x}(t)\rightarrow\textit{BR}_{\mathcal{K}}\left(\bm{p}(t)\right)$.
\end{theorem}

\begin{proof}
    We first prove that $\bm{p}(t)^T\dot{\bm{x}}(t)\rightarrow0$. By integration by parts and (\ref{eq:CCW}), we obtain
    \begin{align}
        \int_0^\infty \bm{p}(t)^T\dot{\bm{x}}(t)\ dt&=\left[\bm{p}(t)^T\bm{x}(t)\right]_0^{\infty}-\int_0^\infty \dot{\bm{p}}(t)^T\bm{x}(t)\ dt\nonumber\\
        &\leq 2\|\bm{p}\|_\infty-\alpha_c<\infty.\label{eq:int}
    \end{align}
    Since $\bm{p}(t)^T\dot{\bm{x}}(t)\geq0$ by (\ref{eq:hier_pc}), the integral in (\ref{eq:int}) exists. By (\ref{eq:for_LaSalle}) and \Cref{ass:conti}, $\bm{p}(t)^T\dot{\bm{x}}(t)$ is Lipschitz continuous. Therefore, by Barbalat's lemma, we obtain that $\bm{p}(t)^T\dot{\bm{x}}(t)\rightarrow0$. Then, by (\ref{eq:for_LaSalle}) and positive correlation, we conclude that $\bm{x}(t)$ tends to be the best response to $\bm{p}(t)$ within $\mathcal{K}$.
\end{proof}

\Cref{thm:PDM} implies, when $\dot{\bm{x}}(t)\rightarrow\bm{0}$, $\bm{x}(t)\rightarrow \textit{NE}_{\mathcal{K}}(F)$, which follows from (\ref{eq:recover}). Even when $\dot{\bm{x}}(t)\not\rightarrow\bm{0}$, we can nevertheless conclude that $\bm{x}(t)$ tends to be the best response to $\bm{p}(t)$ within $\mathcal{K}$. This weaker convergence result compared to \Cref{thm:potential} is since $\bm{p}(t)$ is given by the dynamic payoff model instead of the static one.
\section{Applications to scenarios with constraints}
\label{sec:app}

We demonstrate the results on a routing scenario where travelers use a navigation app that provides three route selection channels, each displaying an estimated travel time based on the average travel time of current users in that channel. Travelers are only informed of the estimated travel times and are unaware of differences between the channels, {\it e.g.}, variations in routing sets or dispatching rules. Further,  travelers select a channel based on the estimated times displayed. Each channel then assigns its users to routes within its routing set according to its dispatching strategy. For simplicity, we assume that all channels share the same routing set, comprising all possible routes. An illustration is provided in \Cref{fig:app}, where $\bm{s}^Y(t)\in\Delta^3$ represents how the channel $Y=A,B,C$, dispatches its users among the three routes at time $t$.

\begin{figure}[!ht]
    \centering
    \includegraphics[width=\linewidth]{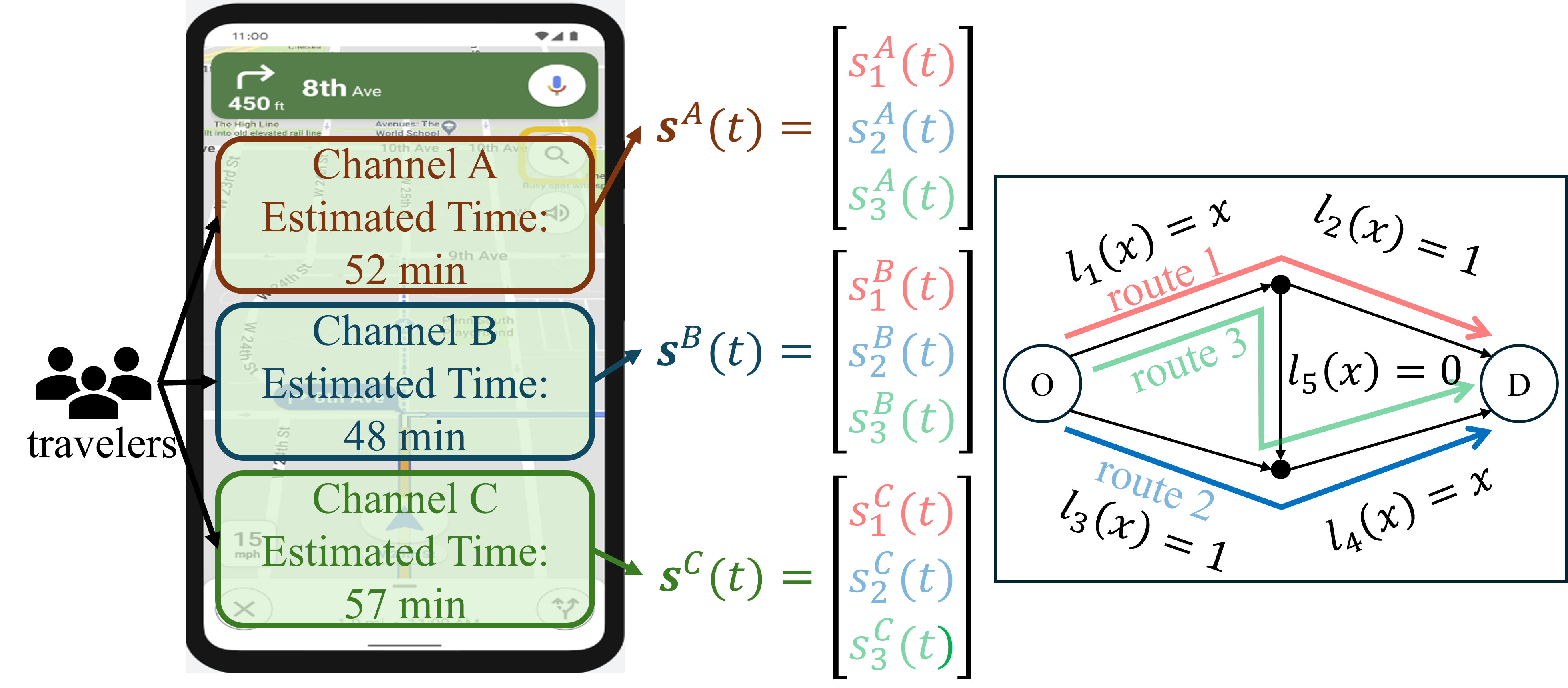}
    \caption{Travelers select channels based on the estimated times and each channel assigns its users to routes.}
    \label{fig:app}
\end{figure}


Let $\bm{x}(t)\in\Delta^3$ be the proportion of travelers on each route, then given the delay functions for links displayed in \Cref{fig:app}, we define the payoff by (negative) travel times:
\begin{align}
    F(\bm{x}(t))=-\begin{bmatrix}
        1+x_1(t)+x_3(t)\\
        1+x_2(t)+x_3(t)\\
        x_1(t)+x_2(t)+2x_3(t)
    \end{bmatrix}.
\end{align}
Then, at time $t$, each channel $Y$ displays the estimated travel time by $-{\bm{s}^Y}(t)^TF(\bm{x}(t))$ and updates its dispatching strategy $\bm{s}^Y(t)$ based on $F(\bm{x}(t))$, for $Y=A,B,C$. This forms a two-layer structure with $\bm{W}^1=\bm{I}_3$ and $\bm{W}^2=\left[\bm{I}_3,\bm{I}_3,\bm{I}_3\right]\in\mathbb{R}^{3\times9}$, where $\bm{I}_3$ is the $3$-by-$3$ identity matrix. Note that $\bm{s}^A(t)$, $\bm{s}^B(t)$, $\bm{s}^C(t)$ in this example are aliases of $\bm{s}^{2,1}(t)$, $\bm{s}^{2,2}(t)$, $\bm{s}^{2,3}(t)$.

We next demonstrate how to leverage this two-layer framework to ensure constraint satisfaction. Consider two scenarios: (i) the proportion of travelers on route 2, $x_2(t)$, must remain below $15\%$, and (ii) the proportion of travelers on route 3, $x_3(t)$, must not exceed $90\%$. The feasible sets, denoted as $\mathcal{D}$, are shown in \Cref{fig:app_sets},
where the vertex $R_3$ is the unique Nash equilibrium. 
In case~(i), the constraint set includes the original Nash equilibrium, whereas in case~(ii), it does not.


We assume decisions made by travelers are unconstrained, i.e., $\mathcal{K}^{1,1}=\Delta^3$. If we design $\mathcal{K}^A,\mathcal{K}^B,\mathcal{K}^C$ such that $\mathcal{K}\subseteq\mathcal{D}$, then $\bm{x}(t)\in\mathcal{D}$ is ensured for all $t$. Since $\bm{W}^1=\bm{I}_3$, $\mathcal{K}$ is a convex combination of the sets $\mathcal{K}^A$, $\mathcal{K}^B$, and $\mathcal{K}^C$. We present design examples in \Cref{fig:app_sets} where $\mathcal{K}=\mathcal{D}$. Let each channel $Y$, update 
$\bm{s}^Y(t)$ by constrained best response dynamics (\ref{eq:CBR}):
\begin{align}
    \dot{\bm{s}}^Y(t)\in\textit{BR}_{\mathcal{K}^Y}\left(F(\bm{x}(t))\right)-\bm{s}^Y(t), \quad t\geq0.
    \label{eq:app_dyn}
\end{align}
Since congestion games are potential games, we conclude by \Cref{thm:potential} that
$\bm{x}(t)$ satisfies the constraint and converges to a Nash equilibrium within $\mathcal{D}$, although travelers are unaware of the constraints and react only to the estimated times.

\begin{figure}[!ht]
    \centering
    \includegraphics[width=.42\columnwidth]{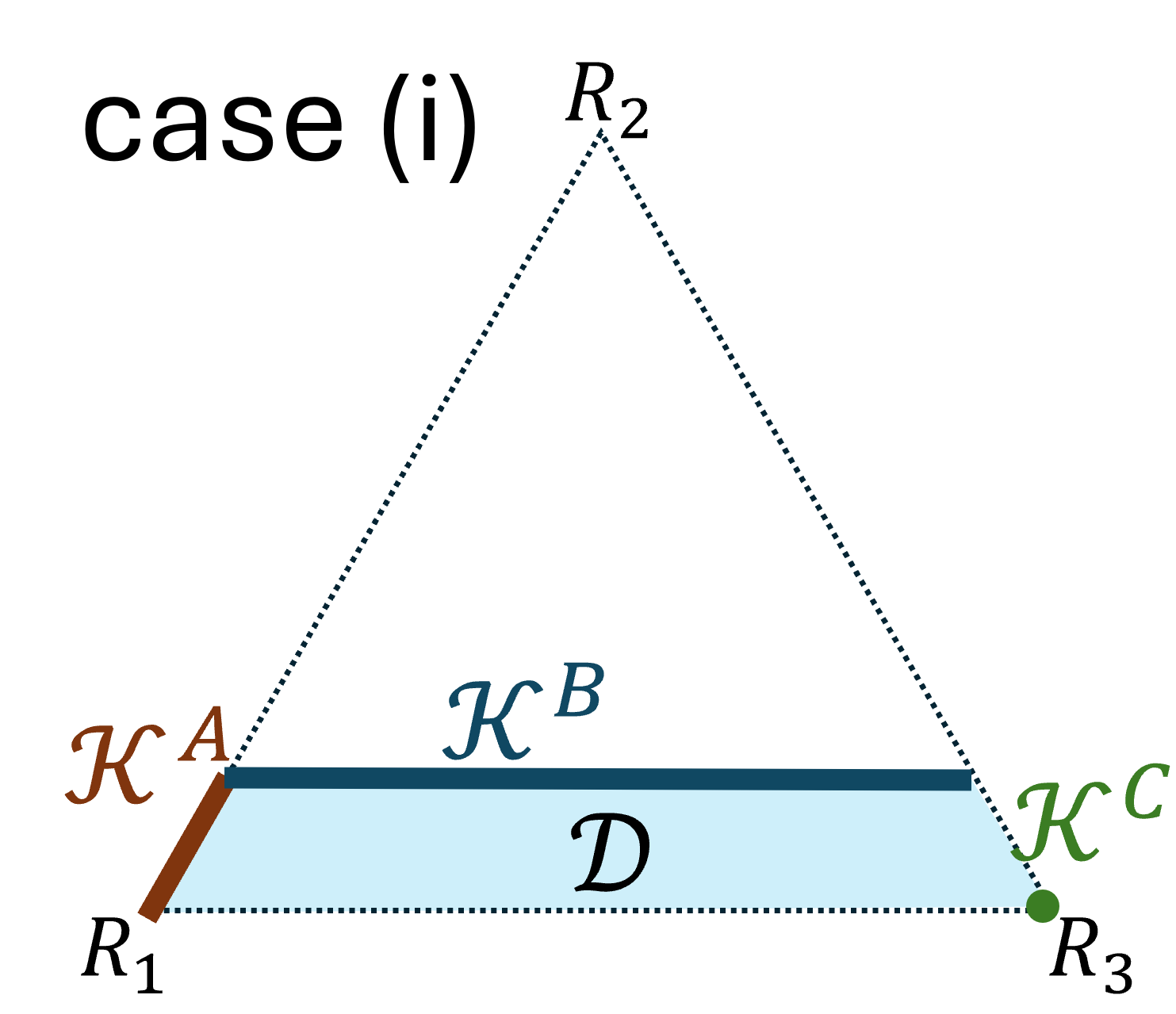}
    \quad\quad
    \includegraphics[width=.42\columnwidth]{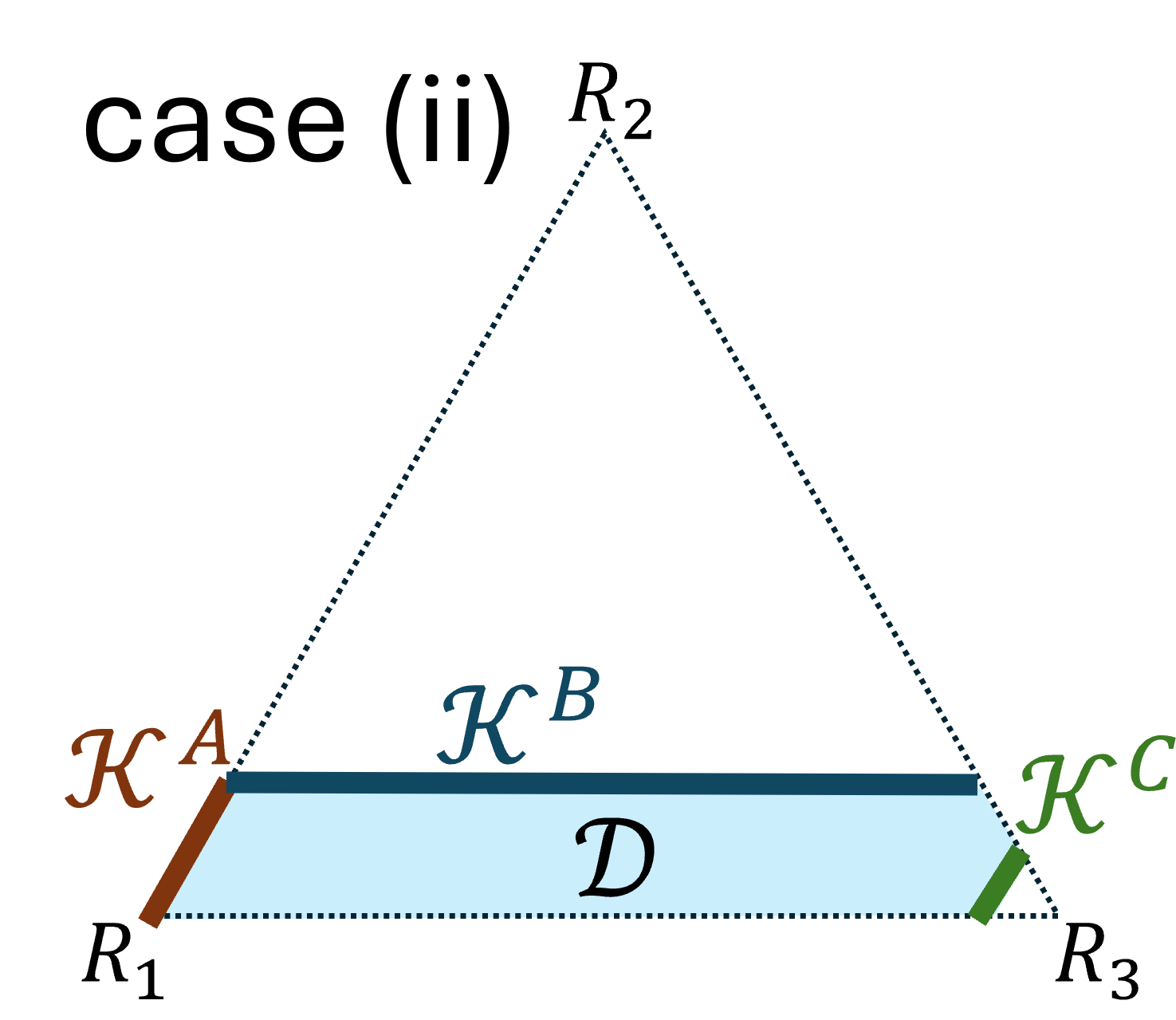}
    \caption{Constructions of $\mathcal{K}^A,\mathcal{K}^B,\mathcal{K}^C$ such that $\mathcal{K}=\mathcal{D}$.}
    \label{fig:app_sets}
\end{figure}

\Cref{fig:traj} demonstrates the simulation results. Thin yellow lines represent the channels' dispatching strategy trajectories, and green dots and red crosses indicate the initial and final dispatching strategies, respectively. The dotted pink line is the route distribution trajectory when travelers select routes directly.
Despite convergence to the Nash equilibrium $R_3$, the constraints are violated during evolution. In contrast, the trajectory for the two-layer case, shown by the thick blue line, converges within $\mathcal{D}$. In case (ii), it converges to a \emph{new} Nash equilibrium. Compared to the original Nash equilibrium $R_3$, flows are dispersed, showcasing the ability to steer the system toward an equilibrium with greater social efficacy.

\begin{figure}[!ht]
    \centering
    \includegraphics[width=.49\columnwidth]{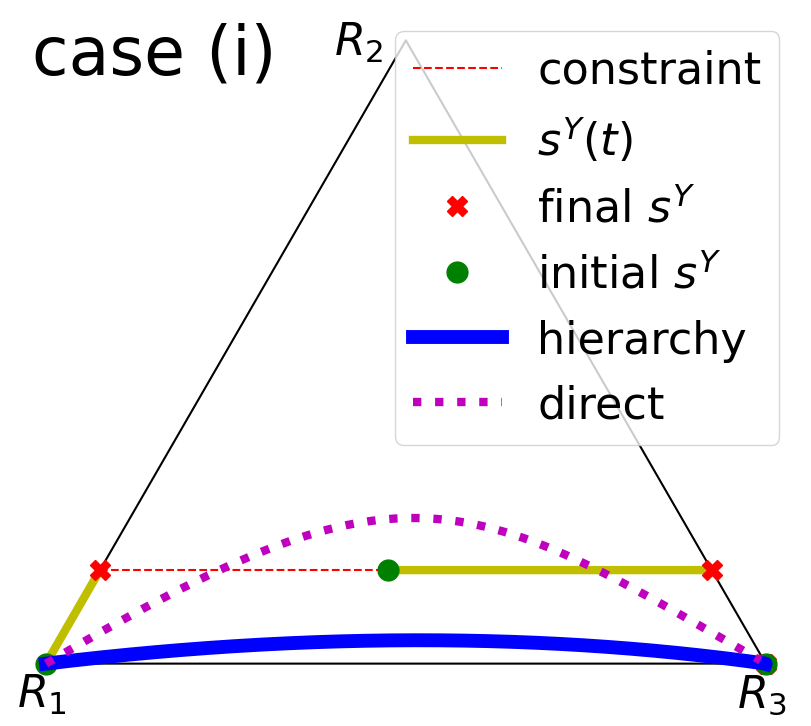}
    \hfill
    \includegraphics[width=.49\columnwidth]{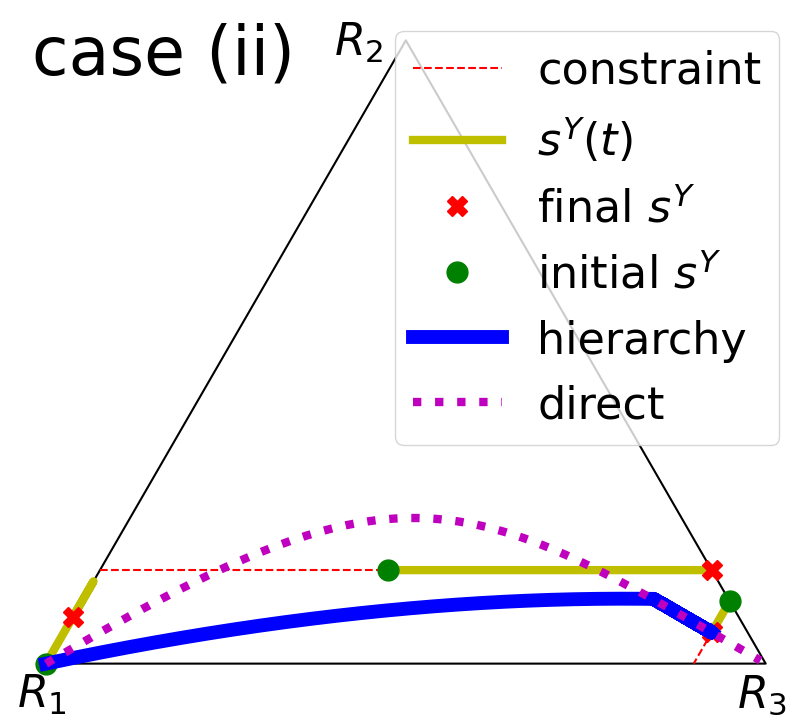}
    \caption{Trajectory plots projected from the simplex to 2D.}
    \label{fig:traj}
\end{figure}

\section{Conclusions}
\label{sec:conclu}

We introduced a hierarchical decision-making framework for population games, extending the classical formulation to multiple decision layers. We characterized equilibrium properties and established convergence under payoff conditions, such as potential games and CCW payoff dynamics. Using these results, we proposed a novel approach for population games with constraints that individuals are unaware of.
Some interesting avenues for future work include investigating broader classes of learning dynamics that satisfy \Cref{ass:dyn} over general convex sets, as well as exploring broader classes of games that converge within the hierarchical framework.

\bibliographystyle{IEEEtran}
\bibliography{reference.bib}

\end{document}